\documentclass[onecolumn,showpacs]{revtex4-2}
\usepackage{graphicx}
\usepackage{dcolumn}
\usepackage[utf8]{inputenc}
\usepackage{amsmath}
\usepackage{amssymb}
\usepackage{color}
\usepackage{hyperref}
\usepackage{scalerel}
\usepackage{undertilde}
\allowdisplaybreaks

\newtheorem{theorem}{Theorem}[section]
\newtheorem{corollary}[theorem]{Corollary}
\newtheorem{proposition}[theorem]{Proposition}

\newtheorem{lemma}[theorem]{Lemma}
\newtheorem{definition}[theorem]{Definition}
\newenvironment{proof}[1][Proof]{\begin{trivlist}
\item[\hskip \labelsep {\bfseries #1}]}{\end{trivlist}}

\newenvironment{remark}[1][Remark]{\begin{trivlist}
\item[\hskip \labelsep {\bfseries #1}]}{\end{trivlist}}

\newcommand{\qed}{\nobreak \ifvmode \relax \else
	\ifdim\lastskip<1.5em \hskip- \lastskip
	\hskip1.5em plus0em minus0.5em \fi \nobreak
	\vrule height0.75em width0.5em depth0.25em\fi}

\begin{document}

\title{Gradient conformal stationarity and the CMC condition in LRS spacetimes}

\author{Gareth \surname{Amery}}
\email{ameryg1@ukzn.ac.za}
\affiliation{Astrophysics Research Center, School of Mathematics, Statistics and Computer Science, University of KwaZulu-Natal, Private Bag X54001, Durban 4000, South Africa}

\author{Peter K S \surname{Dunsby}}
\email{peter.dunsby@uct.ac.za}
\affiliation{Cosmology and Gravity Group, Department of Mathematics and Applied Mathematics, University of Cape Town, Rondebosch 7701, Cape Town, South Africa}
\affiliation{South African Astronomical Observatory, Observatory 7925, Cape Town, South Africa}
\affiliation{Center for Space Research, North-West University, Mahikeng 2745, South Africa}

\author{Abbas M \surname{Sherif}}
\email{abbasmsherif25@gmail.com}
\affiliation{Department of Science Education, Jeju National University, Jeju, 63243, South Korea}

\begin{abstract}
We study the existence of gradient conformal Killing vectors (CKVs) in the class of locally rotationally symmetric (LRS) spacetimes which generalizes spherically symmetric spacetimes, and investigate some implications for the evolutionary character of marginally outer trapped surfaces. We first study existence of gradient CKVs via the obtention of a relationship between the Ricci curvature and the gradient of the divergence of the CKV. This provides an alternative set of equations, for which the integrability condition is obtained, to analyze the existence of gradient CKVs. A uniqueness result is obtained in the case of perfect fluids, where it is demonstrated that the Robertson-Walker solution is the unique perfect fluid solution with a nonvanishing pressure, admitting a timelike gradient CKV. The constant mean curvature condition for LRS spacetimes is also obtained, characterized by three distinct conditions which are specified by a set of three scalars. Linear combinations of these scalars, whose vanishing define the constant mean curvature condition, turn out to be related to the evolutions of null expansions of 2-spheres along their null normal directions. As such, some implications for the existence of black holes and the character of the associated horizons are obtained. It is further shown that dynamical black holes of increasing area, with a non-vanishing heat flux across the horizon, will be in equilibrium, with respect to the frame of the conformal observers.
\end{abstract}

\maketitle


\section{Introduction}


Understanding the role of symmetries in general relativity dates back to the 1920's when Brinkmann examined the usefulness of such transformations to obtain new exact solutions to the field equations \cite{br1,br2}. Since then a wealth of literature has been gathered on this subject (see the references \cite{ehl1,herr1,herr2,herr3,mart1,herr4,aco1,aco2,aco3,mart2} and associated references). Killing and conformal Killing (infinitesimal) symmetries, the action for which the metric is left either unchanged, or simply scaled, are usually assumed. For the conformal Killing case, if the vector field generating the symmetry is timelike, one has a choice of observers, and in which case it is known that the spacetime admitting the symmetry is conformal to a stationary one, i.e. it scales a stationary spacetime by some factor. These spacetimes are termed \textit{conformally stationary} (or CS for short) spacetimes \cite{alias1,rub1}.

In the case that the timelike conformal Killing vector field is gradient, these spacetimes are called gradient CS (GCS) spacetimes. In order for such spacetimes to be stably causal, it is required that the gradient condition on the vector field is a global one. There have been several studies addressing various properties of the spacelike hypersurfaces of the foliation induced by the splitting along the CKV, related to their stability and uniqueness, mostly when the hypersurfaces are maximal (see for example \cite{rub1,rub2} and associated references). A particular subclass of GCS spacetimes is the set of generalized Robertson-Walker (GRW) spacetimes, a warped product spacetime that has been extensively studied for various reasons. (See the following references: \cite{rub3,rub4,rub5,man1}. For an in depth review in the case of GRW spacetimes, the reader is referred to \cite{man2}.)

The relationship between symmetries of spacetimes and trapped and marginally outer trapped surfaces (MOTS) have also been investigated previously in \cite{mm1,ash2}. These studies have primarily focused on the Killing case, where the existence of trapped surfaces, MOTS and certain horizon types were investigated under the Killing symmetry assumption, considering different causal characters of the Killing vector. (In the reference \cite{mm1}, the case of a proper conformal Killing symmetry was briefly commented on, where the averaging over the MOTS, the relationship between the conformal divergence $\varphi$ and the inner product of the conformal Killing vector and the null mean curvature vector of the MOTS was examined.)  More recently, the stability of MOTS contained in an initial data, which admits a Killing symmetry, has been investigated in \cite{ib5}. (Previous studies along such lines, albeit differently approached, were carried out in \cite{mm2}.) Several results demonstrating the instability of exotic MOTS (those with nontrivial topologies), some of which were found in \cite{ib8}, for example, were obtained.

Locally rotationally symmetric (LRS) spacetimes are a class of spacetimes with relevance to astrophysics and cosmology. These spacetimes, which are naturally threaded along observer's four-velocity, admit a preferred spatial direction which is orthogonal to the four-velocity \cite{el1,el2}. These include fluid spacetimes like rotationally symmetric perfect fluids. The subclass, named the LRS II class, for which the two preferred directions have vanishing vorticities, generalizes solutions with spherical symmetry, and is understood to contain the exact spherically symmetric black hole solutions. Locally, the metric of these spacetimes take the form

\begin{eqnarray}\label{fork}
ds^2=-A^2dt^2+ B^2dr^2+ C^2\left(dy^2+D^2dz^2\right)\;,
\end{eqnarray}
where \(A,B\) and \(C\) are \(t\) and \(r\) dependent. The function \(D\), which is a function of \(y\), is parametrized by a number \(k\in\{-1,0,1\}\): \(k=-1\) corresponds to \(\sinh y\), \(k=0\) corresponds to \(y\), \(k=1\) corresponds to \(\sin y\)).

The more general form of the local metric acquires additional terms. We will not need this form here since subsequent analysis will not require the explicit form of the metric, except to mention that in the general case, the vorticities of the preferred directions are generally non-zero. These non-zero vorticities will sparsely appear in the rest of the work, and so will be introduced in the next section.

Conformal symmetries in these spacetimes have been studied variously under several assumptions, and adaptation to the 1+1+2 covariant formalism (see \cite{cc1,cc2}) employed in this work can be found in these recent works \cite{bergh1,singh1,chev1,chev2,as5}. The formalism is quite natural for LRS spacetimes as all vector and tensor quantities can be projected along the two preferred unit directions. Existence of gradient CKV in vacuum and perfect fluid spacetimes was first examined by Daftardar and Dadhich in \cite{nd1}. In the recent work by Koh \textit{et al.}, \cite{as5}, this was generalized to all of LRS spacetimes, with the existence in the LRS II case characterized by a wave-like PDE. The criteria for the existence of gradient CKVs were also discussed. Gradient CKVs have also found a role in generating hidden symmetries in a spacetime. Specifically, these vector fields can be used to generate conformal Killing tensors and Killing tensors. (For example, see \cite{amery1} where the Koutras algorithm, introduced in \cite{koutras1}, was used by Amery and Maharaj to establish the form of the Killing tensor in an Einstein space. Also, more general considerations have been carried out by Rani \textit{et al.}, \cite{rani1}, to generate conformal Killing tensors of the gradient type.)

Our primary aim of this work is to investigate the conditions for an LRS spacetime to be GCS, and the implications for the existence of trapped and MOTS in the spacetime, and how they consequently evolve. The first part of this work will supplement that of \cite{as5}. We, however, approach this problem quite differently. A curvature condition will be obtained which relates to the behavior of the gradient of the divergence of the CKV. The condition provides a set of equations for the potential function for the vector field which we are to analyze. While the use of the conformal Killing equations is purely kinematic, this approach to be employed here clarifies the role of the matter variables (and hence relationship to the energy conditions) in relations to the existence of gradient CKVs.

Furthermore, the presence of a timelike gradient CKV orthogonally splits the timelike region of the spacetime into a stack of spacelike hypersurfaces, on each of which the trace of the extrinsic curvature is constant. Hence, finding the precise condition(s) for this constancy provides necessary condition for the (non)-existence of gradient CKV. This essentially forms the second part of this work. We obtain the constancy condition(s), and from its relationship to the above mentioned curvature condition, we draw some concluding statements about existence. We then use the constant mean curvature condition to study existence/presence of trapped surfaces and MOTS in the spacetime, the evolution of the surface in the case of a MOTS, and how a black hole is observed in the conformal frame. These results, in  some regards, are related to some of the results obtained in \cite{mm1,ash2} for the Killing case.

We organize this paper as follows. In Section \ref{sec2}, we give a brief overview of the 1+1+2 spacetime decomposition, following the standard literature which will be referenced. Section \ref{sec3} introduces the notion of a gradient conformally stationary spacetime. An identity relating the curvature to the divergence of the CKV is obtained, which provides an alternative set of equations, along with the integrability condition, to analyze the existence and character of a gradient CKV. A uniqueness result is also obtained in the case of perfect fluid solutions. In Section \ref{sec4}, we study the constant mean curvature (CMC) condition for the associated spacelike hypersurfaces. This is realized as algebraic constraints on the covariant matter variables, and also implies necessary conditions for the existence of a timelike gradient conformal Killing vector field. Some results for the implications for the existence of black holes, and the character of asociated horizons, are obtained in Section \ref{sec5}. A summary of results is presented in Section \ref{sec6} with potential future directions suggested.


\section{A covariant spacetime decomposition and LRS solutions}\label{sec2}


We briefly introduce the 1+1+2 covariant decomposition approach to be used in our analysis, and specialize to LRS spacetimes only, where 2-tensors and 2-vectors on the 2-space are identically zero. (Discussions of the more general decomposition entailing all of the covariant tensors and vectors can be followed in the references \cite{cc1,cc2}.) These spacetimes are algebraically classed as Petrov type D spacetimes \cite{ellis5}.

In addition to the unit timelike field \(u^a\) which threads a 4-dimensional spacetime as a 1+3 product manifold, the 1+1+2 approach introduces a unit spacelike vector field, denoted \(n^a\), orthogonal to \(u^a\) which decomposes the Riemannian 3-manifold as a 1+2 product manifold. In addition to the `dot' derivative along the timelike \(u^a\) and the metric on the 3-manifold \(h_{ab}=g_{ab}+u_au_b\) induced by the 1+3 splitting, the additional decomposition of the 3-manifold introduces a derivative along \(n^a\), as well as the induced metric on the 2-space, \(N_{ab}=g_{ab}+u_au_b-n_an_b\):

\begin{align*}
\dot{\psi}_{\ \ \ \ a\cdots b}^{f\cdots g}&=u^c\nabla_c\psi_{\ \ \ \ a\cdots b}^{f\cdots g},\\
\hat{\psi}_{\ \ \ \ a\cdots b}^{f\cdots g}&=n^c\nabla_c\psi_{\ \ \ \ a\cdots b}^{f\cdots g},
\end{align*}
for an arbitrary tensor \(\psi_{\ \ \ \ a\cdots b}^{f\cdots g}\), where \(\nabla_a\) is the spacetime covariant derivative. Gradient of a scalar \(\psi\) decomposes as

\begin{align*}
\nabla_a\psi=-\dot{\psi}u_a+\hat{\psi} n_a+\delta_a\psi,
\end{align*}
where \(\delta_a\psi=N^b_{\ a}\nabla_b\psi\) is the 2-surface component of the spacetime gradient of \(\psi\). As we are specializing to LRS spacetimes, the symmetry of these spacetimes imposes the vanishing of the surface term $\delta_a\psi$ of the full covariant derivative of $\psi$.

The energy momentum tensor, under this decomposition, admits the splitting

\begin{align*}
T_{ab}=\rho u_au_b+2Qu_{(a}n_{b)}+ph_{ab}+\Pi\left(n_an_b-\frac{1}{2}N_{ab}\right),
\end{align*}
where \(\rho=T_{ab}u^au^b\) is the energy density, \(3p=T_{ab}h^{ab}\) is the (isotropic) pressure, \(Q=T_{ab}n^au^b\) is the scalar associated to the heat flux vector $q_a=-h^{\ c}_aT_{cd}u^d$, and \(\Pi=T_{ab}n^an^b-p\) encodes the deviation from isotropy. For perfect fluid LRS spacetimes, a subclass that will play a central role in the rest of the paper, one has the vanishing of the flux and anisotrophy variables $Q$ and $\Pi$. Using the field equations, the Ricci tensor and the scalar curvature of the spacetime take the forms

\begin{align}
R_{ab}&=g_1u_au_b+g_2h_{ab}+2Qu_{(a}n_{b)}+\Pi\left(n_an_b-\frac{1}{2}N_{ab}\right),\label{ck12}\\
R&=g_1+3g_2=\rho-3p+4\Lambda,\label{ckt2}
\end{align}
where we have set

\begin{align*}
g_1=\frac{1}{2}\left(\rho+3p-2\Lambda\right);\quad g_2=\frac{1}{2}\left(\rho-p+2\Lambda\right),
\end{align*}
with the quantity $\Lambda$ being the cosmological constant.

Of fundamental importance to almost all calculations in the current formulation is the full covariant derivatives of the unit vector fields \(u^a\) and \(n^a\):

\begin{align}
\nabla_au_b&=-\mathcal{A}u_an_b+\left(\frac{1}{3}\theta+\sigma\right)n_an_b+\frac{1}{2}\left(\frac{2}{3}\theta-\sigma\right)N_{ab}+\Omega\varepsilon_{ab},\label{cd1}\\
\nabla_an_b&=-\mathcal{A}u_au_b+\left(\frac{1}{3}\theta+\sigma\right)n_au_b+\frac{1}{2}\phi N_{ab}+\xi\varepsilon_{ab}.\label{cd2}
\end{align}

In the above equations, \(\mathcal{A}=n_a\dot{u}^a\) is the acceleration of $u^a$, \(\theta=h^{ab}\nabla_au_b=\nabla_au^a\) is the expansion of \(u^a\), \(\sigma=\sigma_{ab}n^an^b\) is the shear scalar (where \(3\sigma_{ab}=3\nabla_{\langle a}u_{b\rangle}=(3h^{\ c}_{(a}h^{\ d}_{b)}-h_{ab}h^{cd})\nabla_cu_d\) is the shear of \(u^a\), with the angle bracket denoting fully projected and trace-free symmetric part), \(\phi=\delta_an^a\) is the sheet expansion, \(\Omega=\omega_an^a\) is the component of the vorticity of \(u^a\) along $n^a$ where $\omega_a=(1/2)\eta_{abcd}u^b\nabla^{[c}u^{d]}$ ($\eta_{abcd}$ is the 4-dimensional alternating tensor and the square brackets denote antisymmetrization), \(\xi=(1/2)\varepsilon^{ab}\delta_an_b\) is the twist of \(n^a\), and \(\varepsilon_{ab}\) is the 2-dimensional alternating tensor.

In general, the dot and hat derivatives do not commute, but obey the following relation \cite{cc1,cc2}:

\begin{eqnarray}
\hat{\dot{\psi}}-\dot{\hat{\psi}}=-\mathcal{A}\dot{\psi}+\left(\frac{1}{3}\theta+\sigma\right)\hat{\psi},\label{comrel}
\end{eqnarray}
for an arbitrary scalar $\psi$. The above commutation relation is useful in obtaining additional constraints and performing consistency checks.

Finally, the field equations can be expressed as first order evolution and propagation equations of the covariant variables (see the appendix).


\section{On gradient CKV in LRS spacetimes}\label{sec3}


In this section we introduce notation and some standard definitions that are of interest to this work. We begin with the definition of conformal symmetries and the gradient case. Existence of gradient CKVs in fluid spacetimes have been discussed \cite{nd1}, and more recently, their existence was studied in the LRS class of spacetimes \cite{as5}. Here, we will consider the properties of gradient CKVs in LRS spacetime more geometrically, purely from the character of the Ricci tensor of the spacetime. Some new results will be provided. Relationship of the energy conditions to the character of the gradient CKV will be discussed and when the CKV is a Ricci principal direction, i.e. when the CKV is an eigenvector of the Ricci tensor, the form and character of the associated eigenvalue is analyzed.

\subsection{Definitions and notations}

A spacetime \(\mathcal{M}\), with metric \(g_{ab}\), admits a conformal symmetry if there exists a vector field \(x^a\) such that

\begin{eqnarray}\label{ck1}
\mathcal{L}_xg_{ab}=\nabla_ax_b+\nabla_bx_a=2\varphi g_{ab},
\end{eqnarray}
where \(\mathcal{L}_x\) is the Lie derivative along \(x^a\), and \(\varphi\) is some smooth function on \(\mathcal{M}\) defined as \(4\varphi=\nabla_ax^a\). (From now on we shall refer to \(\varphi\) as the \textit{conformal divergence} as it captures the divergence of conformal observers along orbits of \(x^a\), when such observers can be defined.) The vector field \(x^a\) is called a Killing vector (KV) if \(\varphi=0\), a homothetic Killing vector (HKV) if \(\varphi\) is constant, and a proper conformal Killing vector (CKV) if \(\varphi\) is non-constant.

We adopt the standard nomenclature for the equations \eqref{ck1}, \textit{conformal Killing equations} (CKE). (In some literature it is sometimes referred to as the `\textit{unphysical Killing equations}').

\begin{definition}
A spacetime \(\mathcal{M}\) is called a \textbf{conformally stationary} (CS) spacetime if it admits a timelike CKV, globally defined. In the case the CKV is globally a gradient, \(\mathcal{M}\) is referred to as a \textbf{gradient conformally stationary} (GCS) spacetime \cite{alias1,rub1}.
\end{definition}

\noindent\textit{Notation fixing:} Going forward, we will simply refer to the function whose gradient is the CKV \(x^a\) as the potential function of \(x^a\), which will be denoted \(\Psi\). The gradient of an arbitrary scalar \(\psi\) will be denoted \(\psi_a=\nabla_a\psi\), and the contraction over an index along the CKV will be notated with a subscript/superscript `\(o\)': \(\psi_{abc}x^c=\psi_{abo}\).

\ \\
\noindent\textit{Energy conditions abbreviations:} Weak energy condition (WEC); Null energy condition (NEC); Strong energy condition (SEC).

\subsection{Some initial results: Gradient CKV, energy conditions and eigenvalue}

The CKV type we consider obeys the LRS symmetry (components are functions of \(t\) and \(r\) coordinates) and takes the form

\begin{eqnarray}\label{ck2}
x^a=\alpha_1 u^a+\alpha_2 n^a,
\end{eqnarray}
where the \(\alpha_i\) are smooth. With this form of the vector field, the CKE \eqref{ck1} can be written down as a set of three covariant first order PDEs in the components \(\alpha_i\), plus a constraint equation:

\begin{align}
\varphi&=\dot{\alpha}_1+\alpha_2\mathcal{A},\label{cmc1}\\
\varphi&=\hat{\alpha}_2+\alpha_1\left(\frac{1}{3}\theta+\sigma\right),\label{cmc2}\\
0&=\dot{\alpha}_2-\hat{\alpha}_1+\alpha_1\mathcal{A}-\alpha_2\left(\frac{1}{3}\theta+\sigma\right),\label{cmc3}\\
2\varphi&=\alpha_1\left(\frac{2}{3}\theta-\sigma\right)+\alpha_2\phi.\label{cmc4}
\end{align}

Additionally, since \(x^a\) is gradient, the bivector

\begin{align}
F_{ab}=\nabla_{[a}x_{b]}=2\left(\dot{\alpha}_2+\alpha_1\mathcal{A}\right)u_{[a}n_{b]}+\left(\alpha_1\Omega+\alpha_2\xi\right)\varepsilon_{ab},\label{biv}
\end{align}
must vanish, and this occurs if and only if, simultaneously,

\begin{align}\label{biv}
\dot{\alpha}_2+\alpha_1\mathcal{A}&=0;\\
\alpha_1\Omega+\alpha_2\xi&=0.
\end{align}

Noting that $x^a=\nabla^a\Psi$, it follows

\begin{eqnarray}
\alpha_1=-\dot\Psi;\quad \alpha_2=\hat{\Psi}.
\end{eqnarray}

In \cite{as5}, it was found that the only solutions admitting a timelike gradient CKV are either the irrotational and twisting ones, i.e. $\Omega=0,\xi\neq0$, or the spacetime must be of the LRS II type, i.e. $\Omega=\xi=0$. (The two other subclasses are LRS I with $\Omega\neq0,\xi=0$ and the LRS III which is the former type with $\Omega=0,\xi\neq0$.) In the latter case, the necessary and sufficient condition for the spacetime to admit a gradient CKV was found. As will later be seen in Section \ref{sec3}, the restriction on the classes of LRS solutions admitting a gradient CKV imposes a certain obstruction to the CMC condition on the hypersurfaces of consideration.

Since an LRS spacetime admitting a timelike gradient CKV is necessarily irrotational, we set $\Omega=0$ in all subsequent expressions.

Let us immediately state the following Lemmas that will be of utility to the rest of this work.

\begin{lemma}\label{leme1}
If a spacetime \(\mathcal{M}\) admits a CKV \(x^a\) with conformal divergence \(\varphi\), then the following hold:

\begin{align}
\Box x_a&=-R_{ao}-2\varphi_a,\label{ck4}\\
\Box\varphi&=-\frac{1}{6}R_o-\frac{1}{3}R\varphi,\label{ck5}
\end{align}
where \(R\) denotes the scalar curvature and $R_o=x^a\nabla_aR$ is the gradient of the scalar curvature along the CKV, with the `box' operator $\Box$ denoting the usual d'Alambertian.
\end{lemma}
\begin{proof}
It is a straightforward exercise, using the Ricci (or curvature) identities for \(x^a\), \(\nabla_a\nabla_bx^c-\nabla_b\nabla_ax^c=R^c_{\ oab}\), and the CKE, to establish the following relation:

\begin{eqnarray}\label{ck3}
\nabla_c \nabla_ax_b = R_{baco} + g_{ab}\varphi_c+g_{cb}\varphi_a-g_{ca}\varphi_b\;.
\end{eqnarray}
This is the CKV analogue of the KV lemma \(\nabla_c\nabla_ax_b=R_{baco}\). Then, \eqref{ck4} follows from \eqref{ck3}, and \eqref{ck5} follows from using the definition of the conformal divergence and making use of the Ricci identities for the CKV.\qed
\end{proof}
\begin{lemma}\label{leme2}
If a spacetime \(\mathcal{M}\) admits a gradient CKV \(x^a\), then, it holds that

\begin{eqnarray}\label{ck8}
R_{ao}=-3\varphi_a.
\end{eqnarray}
\end{lemma}
\begin{proof}
Suppose \(\mathcal{M}\) admits a gradient CKV \(x^a\). Then, of course,

\begin{eqnarray}\label{ck6}
\nabla_ax_b=\varphi g_{ab}.
\end{eqnarray}
Thus,

\begin{eqnarray}\label{ck7}
\Box x_a=\varphi_a,
\end{eqnarray}
which upon inserting into \eqref{ck4} gives \eqref{ck8}.\qed
\end{proof}

This Lemma \ref{leme1} appears in the literature elsewhere (see for example \cite{feng1} where in the definition \eqref{ck1} the author has instead used the scaling $2\varphi\rightarrow\varphi$). The Lemma \ref{leme2}, however, is new (as far as we are aware) and while apparently obvious, clearly has geometric consequences, by informing us how the spacetime Ricci curvature is related to the behavior of the conformal observers. This has consequences for the character of the CKV as well as the geometry of the ambient spacetime and the hypersurfaces of interest. The particular relevance with regards to the character and existence of the CKV will be laid bare in the sequel.

\subsubsection{Some quick comments on the relation of Lemma \ref{leme2}}

Immediately, by Lemma \ref{leme2} we see that the conformal divergence relates to the potential function $\Psi$ via the following set of first order PDEs:

\begin{align}
g_1\dot{\Psi}+Q\hat{\Psi}&=3\dot{\varphi},\label{cdi1}\\
-\left[Q\dot{\Psi}+\left(g_2+\Pi\right)\hat{\Psi}\right]&=3\hat{\varphi},\label{cdi2}
\end{align}
where $g_1$ and $g_2$ are the matter functions appearing in the expression for the Ricci tensor \eqref{ckt2}. If the above pair of equations does not admit a solution $\Psi$ on the space of functions in a LRS spacetime, the spacetime cannot admit a gradient CKV. On the other hand, it is seen that given a CKV with conformal divergence $\varphi$, the existence of a solution \(\Psi\) to the above implies the CKV is gradient.

We make some immediate observations from the pair of equations \eqref{cdi1} and \eqref{cdi2} and obtain some results.

Let us consider cases where the following conditions hold:

\begin{eqnarray}
g_1\neq0;\quad g_1\left(g_2+\Pi\right)\neq Q^2,\label{imp1}
\end{eqnarray}
so that

\begin{eqnarray}
\hat{\Psi}=-3\left(\frac{Q\dot{\varphi}+g_1\hat{\varphi}}{g_1\left(g_2+\Pi\right)-Q^2}\right).\label{cdi5}
\end{eqnarray}
Then, the following statements are in order.

\begin{enumerate}

\item Any gradient CKV $x^a$ must be proper for otherwise, $x^a$ is trivial: If $\varphi$ is constant, then $\hat{\Psi}=0$ implies $\dot{\Psi}=0$ from \eqref{cdi1} since $g_1\neq0$.

\item For $\varphi=\varphi(r)\neq0$ (note that in this case $\hat{\Psi}\neq0$ since $g_1\neq0$), then, in the case of vanishing heat flux, $x^a$ must align with $n^a$.  Indeed, this particular case is not of interest to this work since $x^a$ is spacelike.

\item If $\varphi=\varphi(t)$, then a perfect fluid matter type (assuming a vanishing cosmological constant) with $\rho\neq p$ would necessitate that $x^a$ aligns with $u^a$. In this case, one can in fact write down explicitly the gradient CKV for the spacetime as

\begin{eqnarray}
x^a=-3\left(\frac{\dot{\varphi}}{g_1}\right)u^a.
\end{eqnarray}
A well known example is the gradient CKV of the Robertson-Walker metric. In fact, without assuming a $t$-only dependence of $\varphi$, this result will be true under the condition of positive energy density.

\end{enumerate}

Suppose $x^a$ is a Ricci principal direction (this will always be the case when the timelike criterion is imposed on $x^a$, as will be discussed later), and denote by $\eta$ the eigenvalue associated to $x^a$, i.e.

\begin{eqnarray*}
R_{ab}x^a=\eta x_b.\label{eigv}
\end{eqnarray*}
We analyze this as the two distinguished cases $Q=0$ and $Q\neq0$.

Let us begin with the former, $Q=0$. Firstly, the right hand sides of \eqref{cdi1} and \eqref{cdi2} are $-\eta\dot{\Psi}$ and $-\eta\hat{\Psi}$, respectively. Thus, for $Q=0$, the pair can be reduced to the following system:

\begin{align}
\left(g_1+\eta\right)\dot{\Psi}&=0,\label{eig1}\\
\left(g_2+\Pi-\eta\right)\hat{\Psi}&=0.\label{eig2}
\end{align}
With interest in timelike gradient CKV, $\dot{\Psi}\neq0$, and hence,

\begin{eqnarray}
g_1+\eta=0,\label{eig3}
\end{eqnarray}
thereby giving the associated eigenvalue as

\begin{eqnarray}
\eta=-g_1.\label{eig4}
\end{eqnarray}
The above gives a criterion for the Ricci tensor of a GCS LRS spacetime with vanishing heat flux to have a zero eigenvalue: the WEC marginally holds, i.e. $g_1=0$. On the other hand, we can make the following statement: \textit{The Ricci tensor of a non-radiating GCS LRS spacetime strictly obeying the SEC has at least one negative eigenvalue.}

We will now show that for a positive energy density solution (which is assumed throughout) with a nonvanishing $p$, if $\Pi=0$, then $x^a$ must lie along $u^a$. More precisely,

\begin{proposition}\label{propoor1}
Any gradient CKV of a GCS perfect fluid LRS spacetime with nonvanishing pressure, a positive energy density, and a non-negative cosmological constant, and obeying the SEC, must lie along $u^a$.
\end{proposition}
\begin{proof}
Assume $\hat{\Psi}\neq0$.  It follows that $\rho+p=0$ from \eqref{eig2}. It then follows that

\begin{eqnarray*}
\rho+3p-2\Lambda=-2\left(\rho+\Lambda\right)<0,
\end{eqnarray*}
thereby failing the SEC. Hence, $\hat{\Psi}=0$.\qed
\end{proof}

Of course, as $\hat{\Psi}=\alpha_2=0$ identically, $\mathcal{A}$ must vanish by \eqref{biv}. Furthermore, it is easily checked from the CKE that the spacetime is shear-free. Then, from the field equations \eqref{evo1}, \eqref{evo3}, this would impose that the spacetime is necessarily conformally flat, i.e. $\mathcal{E}=0$. It therefore follows that Proposition \ref{propoor1} implies the following result:

\begin{theorem}\label{theo1}
The only GCS perfect fluid LRS II spacetime with a nonvanishing pressure and non-negative cosmological constant is the Robertson-Walker type solution.
\end{theorem}

As will later be seen, the above result can be strengthened, where we can do away with imposing a condition on the pressure and the existence of timelike gradient CKV in inhomogenous perfect fluids will be ruled out.

\subsubsection{Integrability condition}

We will now derive integrability conditions for the system \eqref{cdi1} and \eqref{cdi2}.

Let us carry out a consistency check for the set of equations \eqref{cdi1} and \eqref{cdi2}. Take the ``hat'' derivative of \eqref{cdi1} and the ``dot'' derivative of \eqref{cdi2} and take the difference of the resulting equations. Then, compare this difference to the equation obtained when the commutation relation \eqref{comrel} is applied to the conformal divergence $\varphi$ and find

\begin{eqnarray}
0=Q\left(\ddot{\Psi}+\hat{\hat{\Psi}}\right)+g_1\hat{\dot{\Psi}}+\left(g_2+\Pi\right)\dot{\hat{\Psi}}+G_1\dot{\Psi}+G_2\hat{\Psi},\label{wws1}
\end{eqnarray}
where we have defined the scalars

\begin{align*}
G_1&=\hat{g}_1+\mathcal{A}g_1+\dot{Q}+\left(\frac{1}{3}\theta+\sigma\right)Q,\\
G_2&=\left(\dot{g}_2+\dot{\Pi}\right)+\left(\frac{1}{3}\theta+\sigma\right)\left(g_2+\Pi\right)+\hat{Q}+\mathcal{A}Q.
\end{align*}
For now we begin quite generally, keeping in both the quantities $\Pi$ and $Q$, for which we can always specialize. Now, it was established in \cite{as5} that the potential function $\Psi$ of the gradient CKV obeys the wave-like PDE 

\begin{eqnarray}
-\ddot{\Psi}+\hat{\hat{\Psi}}+\left(2\sigma-\frac{1}{3}\theta\right)\dot{\Psi}+\left(\mathcal{A}-\phi\right)\hat{\Psi}=0.\label{wws2}
\end{eqnarray}
(The above equation is obtained directly from the covariant conformal Killing equations, using the gradient property of the CKV. See the reference \cite{as5} for the derivation and explanation.) Compactly, the above has the form $\Box\Psi=4\varphi$. Comparing \eqref{wws2} and \eqref{wws1} then gives

\begin{eqnarray}
0=2Q\ddot{\Psi}+\left(\rho+p+\Pi\right)\hat{\dot{\Psi}}+F_1\dot{\Psi}+F_2\hat{\Psi},\label{wws3}
\end{eqnarray}
where we have defined

\begin{align*}
F_1&=\hat{g}_1+\mathcal{A}\left(\rho+p+\Pi\right)+\left[\dot{Q}+\left(\frac{2}{3}\theta-\sigma\right)Q\right],\\
F_2&=\left(g_2+\Pi\right)^{\cdot{\ }}+\left(\hat{Q}+\phi Q\right).
\end{align*}
The equation \eqref{wws3} provides the integrability condition for the set of equations \eqref{cdi1} and \eqref{cdi2}.

If one were to restrict to the perfect fluid case, and note that $\hat{\Psi}=0$ implies $\mathcal{A}=0$ which implies $\hat{\dot{\Psi}}=0$,  we have that $F_1=\hat{g}_1$ and it therefore follows that the condition \eqref{wws3} simplifies to

\begin{eqnarray}
\hat{g}_1\dot{\Psi}=0.\label{wws4}
\end{eqnarray}
Hence, as $\dot{\Psi}$ is nonzero, we have that $g_1$ is constant along the spatial congruence: $\hat{g}_1=0$ gives

\begin{eqnarray}
\hat{\rho}+3\hat{p}=0.\label{wws5}
\end{eqnarray}
We saw that conformal flatness is necessary, and from the field equations one checks that the quantities $\rho$ and $p$ are constant along $n^a$, thereby satisfying the condition of equation \eqref{wws5}.

Now, let us consider those cases with a nonvanishing $Q$. Again, with interest in the timelike case, since $\dot{\Psi}\neq0$, it is easily checked that $Q$ obeys

\begin{eqnarray}
Q^2=\left(g_1+\eta\right)\left(g_2+\Pi-\eta\right).\label{eig10}
\end{eqnarray}
(To obtain the above, note that $x^a$ is an eigenvector of the Ricci tensor with eigenvalue $\eta$, so that $\eta x_a=-3\varphi_a$ by \eqref{ck8}. Comparing components we have $3\dot{\varphi}=-\eta\dot{\Psi}$ and $3\hat{\varphi}=-\eta\hat{\Psi}$. Substitute this into \eqref{cdi1} and \eqref{cdi2}, rewrite \eqref{cdi1} in terms of $\dot{\Psi}$ as $\dot{\Psi}\neq0$, and substitute into \eqref{cdi2}.) Since the right hand side of \eqref{eig10} must be positive, expanding out the right hand side of \eqref{eig10} gives

\begin{eqnarray*}
g_1\left(g_2+\Pi\right)-\eta\left(g_1-\left(g_2+\Pi\right)\right)>\eta^2\longrightarrow g_1\left(g_2+\Pi\right)-\eta\left(g_1-\left(g_2+\Pi\right)\right)>0,
\end{eqnarray*}
so that one obtains the following upper bound on the eigenvalue:

\begin{eqnarray}
\eta<\frac{g_1\left(g_2+\Pi\right)}{g_1-g_2-\Pi}.\label{eig11}
\end{eqnarray}
We note that uniqueness of the system \eqref{cdi1} and \eqref{cdi2} requires $\eta\neq0$. If this is the case, as long as $\rho\geq p$, then for the case of a nonnegative cosmological constant, if $\eta\neq g_1-\left(g_2+\Pi\right)$, the conditions \eqref{imp1} will always hold.

We now return to the condition of \eqref{wws3}. In fact, the condition \eqref{wws3} can be brought to a parabolic form using the first equation of \eqref{biv} and \eqref{cmc3}:

\begin{eqnarray}
2Q\ddot{\Psi}+F_1\dot{\Psi}+\bar F_2\hat{\Psi}=0,\label{eig12}
\end{eqnarray}
with the introduction of the scalar

\begin{eqnarray*}
\bar F_2=F_2+\left(\rho+p+\Pi\right)\left(\frac{1}{3}\theta+\sigma\right).
\end{eqnarray*}

With the above picture \eqref{eig12}, for example, if one seeks a gradient CKV with $u^a$ component a constant of motion along $u^a$, then, one simply needs to check that a solution to the first order equation exists:

\begin{eqnarray}
F_1\dot{\Psi}+\bar F_2\hat{\Psi}=0.\label{eig13}
\end{eqnarray}
The above equation can further be analyzed by noting that is should be consistently propagated along $x^a$.

Another consideration is to attempt to look for a $u^a$-directed $x^a$, i.e. $\hat{\Psi}=0$, for the radiating case $Q\neq0$. Of course then the equation \eqref{eig12} reduces to

\begin{eqnarray}
2Q\ddot{\Psi}+F_1\dot{\Psi}=0.\label{eig14}
\end{eqnarray}
If both the quantities $F_1$ and $Q$ have no zeros, then, the above equation has a solution. On the other hand, if either one of quantities $F_1$ or $Q$ has a zero somewhere and the other has not, as we are only considering the proper conformal Killing case, the equation \eqref{eig14} will not admit a solution. As will later be seen, a gradient conformal observer experiences no radiation even in a radiating spacetime. Hence, as \eqref{eig14} should be consistently propagated along $x^a$, \eqref{eig14} admits a solution if and only if $Q$ and $F_1$ vanish simultaneously. This then would further impose the requirement that

\begin{eqnarray}
\hat{g}_1+\mathcal{A}\left(\rho+p+\Pi\right)=0.\label{eig15}
\end{eqnarray}

Thus, even in an accelerating spacetime that is simultaneously radiating, it is seen that we have a vanishing condition similar to the Robertson-Walker type case \eqref{wws5}, provided that $\rho+p+\Pi$ vanishes. As will be discussed later, the particular gradient CKV determines the foliation, with such specificity captured in one of the constant mean curvature conditions. One of these conditions is in fact specified by this linear relationship of the matter variables $\rho,p$ and $\Pi$.


\section{The constant mean curvature condition}\label{sec4}


We now provide an analysis of the CMC condition in context of some LRS solutions, and the constraint imposed by the conditions on these solutions. Some implications for the existence/nonexistence results of gradient CKVs are discussed.

\subsection{Spacelike hypersurfaces and mean curvature}

Given the existence of a timelike gradient CKV in a spacetime, the global Frobenius theorem guarantees a foliation \(\mathcal{W}_x\) of the region of the spacetime by spacelike hypersurfaces, the leaves of \(\mathcal{W}_x\), to which the CKV \(x^a\) is orthogonal. This follows from the fact that as \(x^a\) is a gradient, its dual is a closed form and hence, the associated orthogonal distribution to the CKV is integrable.

Our interest is in the leaves of the foliation \(\mathcal{W}_x\), in GCS LRS spacetimes. The unit normal to the leaves is the normalized timelike gradient CKV \(\tilde{x}^a=fx^a\) (defining the conformal observers, with \(f=1/\sqrt{-x_bx^b}\)). On such a spacelike hypersurface is therefore induced the Riemannian metric

\begin{eqnarray}\label{sh1}
z_{ab}=g_{ab}+\tilde{x}_a\tilde{x}_b.
\end{eqnarray}
From now on, we will label such a hypersurface as \(\mathcal{T}\). The projected covariant derivative on \(\mathcal{T}\) is then simply

\begin{eqnarray*}
\mathcal{D}_a=z^b_{\ a}\nabla_b.
\end{eqnarray*}

Using the fact that both \(f\) (actually the norm \(x_ax^a\)) and \(\varphi\) are constant on \(\mathcal{T}\), we have the second fundamental form computed as

\begin{eqnarray}\label{sh2}
\chi_{ab}=\frac{1}{2}\bar{\mathcal{L}}_{\tilde{x}}z_{ab}=f\varphi z_{ab},
\end{eqnarray}
with the bar over the Lie derivative operator indicating Lie derivative with respect to the connection \(\mathcal{D}_a\). This implies that these hypersurfaces are totally umbilical, i.e., the second fundamental form is pure trace. The mean curvature is therefore

\begin{eqnarray}\label{sh3}
\chi=-\frac{1}{3}z^{ab}\chi_{ab}=-f\varphi,
\end{eqnarray}
which is constant on \(\mathcal{T}\) as both \(f\) and \(\varphi\) are, i.e. \(\mathcal{T}\) is a CMC hypersurface. Of course, \(\mathcal{T}\) is maximal (\(\chi=0\)) if and only if \(x^a\) is a true KV. In other words, $\mathcal{T}$ is totally geodesic if and only if $x^a$ is a true KV. Indeed its obvious that $\mathcal{T}$ is time symmetric in such a case as $\chi_{ab}$ is identically zero.

Note that the constancy of $\varphi$ on $\mathcal{T}$ implies its gradient $\varphi^a$ is normal to $\mathcal{T}$, and hence $\varphi^a\propto x^a$. From \eqref{ck8} this gives that $x^a$ is an eigendirection for the Ricci tensor as was earlier mentioned in Section \ref{sec3}.

The constancy of the mean curvature essentially constrains the character of the gradient CKV, and therefore in essence, the ambient spacetime itself. For our purpose, we will seek to identify the constraints allowing for the CMC condition in GCS LRS spacetimes. This also provides necessary (and sufficient, in some cases) conditions for the spacetimes to admit a gradient CKV.


\subsection{The CMC criteria and geometric constraints}


We now obtain the constant mean curvature conditions for the LRS class of spacetimes. These conditions will then be discussed in detail, with some implications for the existence of a gradient CKV.

For an arbitrary scalar \(\psi\) in an LRS solution, the \(\mathcal{T}\)-gradient is expressed as

\begin{align}
\mathcal{D}_a\psi&=\nabla_a\psi+f^2\left(x^b\nabla_b\psi\right)x_a\nonumber\\
&=\left[\left(f^2\alpha_1^2-1\right)\dot{\psi}+f^2\alpha_1\alpha_2\hat{\psi}\right]u_a+\left[\left(f^2\alpha_2^2+1\right)\hat{\psi}+f^2\alpha_1\alpha_2\dot{\psi}\right]n_a.
\end{align}
It follows that the condition for the constancy of the scalar \(\psi\) on the hypersurface \(\mathcal{T}\) can be reduced to the following relation on $\mathcal{T}$:

\begin{eqnarray}\label{cme2}
\alpha_2\dot{\psi}+\alpha_1\hat{\psi}=0.
\end{eqnarray}
To see this, note that the $u^a$ and $n^a$ components of (41) must simultaneously vanish as they are orthogonal directions. Then noting $f^2=-1/(x_ax^a)=1/(\alpha_1^2-\alpha_2^2)$, the $u^a$ and $n^a$ components simplify respectively to

\begin{align*}
-u^a\mathcal{D}_a\psi&=\left(f^2\alpha_1^2-1\right)\dot{\psi}+f^2\alpha_1\alpha_2\hat{\psi}=\frac{\alpha_2}{\left(\alpha_1^2-\alpha_2^2\right)}\left(\alpha_1\hat{\psi}+\alpha_2\dot{\psi}\right),\\
n^a\mathcal{D}_a\psi&=\left(f^2\alpha_2^2+1\right)\hat{\psi}+f^2\alpha_1\alpha_2\dot{\psi}=\frac{\alpha_1}{\left(\alpha_1^2-\alpha_2^2\right)}\left(\alpha_1\hat{\psi}+\alpha_2\dot{\psi}\right).
\end{align*}
As $\alpha_1\neq0$ since $x^a$ is timelike, the required vanishing of the above leads to (42).

Obviously, we do not consider the case \(\alpha_1=0\) as the CKV is spacelike. Clearly, if we consider the case of \(\alpha_2=0\), one has that \(\hat{\psi}=0\).

If on the other hand we consider the case with \(\alpha_i\neq0\), then an obvious nontrivial solution is \(\dot{\psi}=\hat{\psi}=0\). And if neither \(\dot{\psi}\) nor \(\hat{\psi}\) is zero, then it follows that a constant \(\psi\) on \(\mathcal{T}\) must obey

\begin{eqnarray}\label{cme3}
\frac{\dot{\psi}}{\hat{\psi}}=-\frac{\alpha_1}{\alpha_2}=\frac{\dot{\Psi}}{\hat{\Psi}}\quad\mbox{which implies}\quad \frac{\dot{\psi}^2}{\hat{\psi}^2}>1,
\end{eqnarray}
using the timelike character of $x^a$.

In summary, the constancy of a scalar $\psi$ on $\mathcal{T}$ is given by exactly one of the following conditions:

\begin{itemize}

\item $\alpha_2=\hat{\psi}=0$ and $\dot{\psi}\neq0$;

\item $\hat{\psi}=\dot{\psi}=0$; and

\item $\hat{\psi}\neq0$ and $\dot{\psi}\neq0$.

\end{itemize}

We note that the last two conditions above are only applicable to LRS II solutions, since for the irrotational and twisting case, $\alpha_2$ is necessarily zero (see the reference \cite{as5} for discussions).

Now, let us return to the equation \eqref{sh3} and state the following

\begin{lemma}
Let \(x^a\) be a timelike gradient CKV in a LRS spacetime \(M\) with vanishing cosmological constant, and let \(\mathcal{T}\) be a hypersurface in \(M\) to which integral curves of \(x^a\) are orthogonal. Then, the evolution and propagation equations of the mean curvature $\chi$ of \(\mathcal{T}\) is given as

\begin{align}
\dot{\chi}&=\frac{\left[2\alpha_1\varphi^2-\left(\alpha_1^2-\alpha_2^2\right)\left(\alpha_1X_1+\alpha_2 X_2\right)\right]}{2\left(\alpha_1^2-\alpha_2^2\right)^{3/2}},\label{cme4}\\
\hat{\chi}&=-\frac{\left[2\alpha_2\varphi^2+\left(\alpha_1^2-\alpha_2^2\right)\left(\alpha_1X_2+\alpha_2 X_3\right)\right]}{2\left(\alpha_1^2-\alpha_2^2\right)^{3/2}},\label{cme5}
\end{align}
where we have defined the scalars

\begin{align*}
X_1&=-\frac{2}{3}g_1+\mathcal{E}-\frac{1}{2}\Pi,\\
X_2&=Q,\\
X_3&=2\xi^2-\frac{2}{3}\rho-\mathcal{E}-\frac{1}{2}\Pi,
\end{align*}
and where the scalar \(\mathcal{E}=E_{ab}e^ae^b\) is the electric part of the Weyl tensor.
\end{lemma}
\begin{proof}
We compute

\begin{eqnarray*}
\dot{f}=-\frac{\alpha_1\dot{\alpha}_1-\alpha_2\dot{\alpha}_2}{\left(\alpha_1^2-\alpha_2^2\right)^{3/2}},
\end{eqnarray*}
so that upon using \eqref{cmc1} and the first equation of \eqref{biv} to respectively substitute for $\dot{\alpha}_1$ and $\dot{\alpha}_2$ we have

\begin{eqnarray*}
\dot{f}=-\frac{\alpha_1\varphi}{\left(\alpha_1^2-\alpha_2^2\right)^{3/2}}.
\end{eqnarray*}

Taking the ``dot'' derivative of \eqref{cmc4}, we can similarly use \eqref{cmc1} and the first equation of \eqref{biv} to substitute for $\dot{\alpha}_1$ and $\dot{\alpha}_2$ and obtain

\begin{eqnarray*}
2\dot{\varphi}=\alpha_1\left[\left(\frac{2}{3}\dot{\theta}-\dot{\sigma}\right)-\mathcal{A}\phi\right]+\alpha_2\left[\dot{\phi}-\mathcal{A}\left(\frac{2}{3}\theta-\sigma\right)\right]+\varphi\left(\frac{2}{3}\theta-\sigma\right).
\end{eqnarray*}
And upon using \eqref{evo1} and \eqref{evo100} to substitute for $(2/3)\dot{\theta}-\dot{\sigma}$ and $\dot{\phi}$ respectively we have

\begin{eqnarray*}
2\dot{\varphi}=\alpha_1\left[-\frac{1}{2}\left(\frac{2}{3}\theta-\sigma\right)^2+X_1\right]+\alpha_2\left[-\frac{1}{2}\phi\left(\frac{2}{3}\theta-\sigma\right)+X_2\right]+\varphi\left(\frac{2}{3}\theta-\sigma\right),
\end{eqnarray*}
with $X_1$ and $X_2$ as defined in the statement of the lemma. We can then use the definition of $\varphi$ in \eqref{cmc4} to establish

\begin{eqnarray*}
2\dot{\varphi}=\alpha_1X_1+\alpha_2X_2.
\end{eqnarray*}
Therefore, substituting $\dot{f}$ and $\dot{\varphi}$ into $\dot{\chi}=-\dot{f}\varphi-f\dot{\varphi}$ gives the evolution of $\chi$ \eqref{cme4}.

Similarly we have 

\begin{eqnarray*}
\hat{f}=-\frac{\alpha_1\hat{\alpha}_1-\alpha_2\hat{\alpha}_2}{\left(\alpha_1^2-\alpha_2^2\right)^{3/2}},
\end{eqnarray*}
Using \eqref{cmc2} and the equation resulting from substituting the first equation of \eqref{biv} into \eqref{cmc3}, to substitute respectively for $\hat{\alpha}_1$ and $\hat{\alpha}_2$ we have

\begin{eqnarray*}
\hat{f}=\frac{\alpha_2\varphi}{\left(\alpha_1^2-\alpha_2^2\right)^{3/2}}.
\end{eqnarray*}

We can take the ``hat'' derivative of \eqref{cmc4} and use the same substitutions for $\hat{\alpha}_1$ and $\hat{\alpha}_2$  to obtain

\begin{eqnarray*}
2\hat{\varphi}=\alpha_1\left[\left(\frac{2}{3}\hat{\theta}-\hat{\sigma}\right)-\left(\frac{2}{3}\theta-\sigma\right)\phi\right]+\alpha_2\left[\hat{\phi}-\left(\frac{2}{3}\theta-\sigma\right)\left(\frac{2}{3}\theta-\sigma\right)\right]+\varphi\phi.
\end{eqnarray*}
And upon using \eqref{evo2} and \eqref{evo200} to substitute for $(2/3)\hat{\theta}-\hat{\sigma}$ and $\hat{\phi}$ respectively we have

\begin{eqnarray*}
2\hat{\varphi}=\alpha_1\left[-\frac{1}{2}\phi\sigma-\frac{1}{3}\phi\theta+X_2\right]+\alpha_2\left[-\frac{1}{2}\phi^2+X_3\right]+\varphi\phi.
\end{eqnarray*}
Again, we can use the definition of $\varphi$ in \eqref{cmc4} to establish

\begin{eqnarray*}
2\hat{\varphi}=\alpha_1X_2+\alpha_2X_3.
\end{eqnarray*}
And upon substituting $\hat{f}$ and $\hat{\varphi}$ into $\hat{\chi}=-\hat{f}\varphi-f\hat{\varphi}$ gives the propagation of $\chi$ \eqref{cme5}.\qed
\end{proof}

The introduction of the scalars \(X_j\), for \(j\in\{1,2,3\}\), is actually not by coincidence. While they conveniently simplify the expressions for the CMC condition, as will be seen shortly, the scalars \(X_j\) are related to the evolution of null expansions of constant \(t\) and \(r\) surfaces in LRS spacetimes, along ingoing and outgoing null paths, and encode implications for the existence of black holes. (We knew the forms of the evolutions of the null expansions of surfaces before hand, and introduced the scalars purposefully so as to relate the CMC condition to the black hole horizons evolution.) We shall return to this later, but before then, let us examine each case of solutions to \eqref{cme2} for the mean curvature \(\chi\).

The CMC condition can effectively be cast as constraints on the scalars \(X_j\). More specifically, for a GCS LRS spacetime, the CMC condition for the spacelike leaves of the foliation, induced by the distribution of the gradient CKV is given by either one of the following cases.

\begin{enumerate}

\item \(\alpha_2=\hat{\chi}=0; \dot{\chi}\neq0\):

\ \\
In this case, we know that \(\hat{\chi}=-X_2/2=0\), and since the spacetime is necessarily shear-free (as is easily seen from the CKE). Combining with the requirement that $\dot{\chi}\neq0$, the CMC condition reduces to the pair of constraints on \(\mathcal{T}\):

\begin{align*}
X_1&\neq \frac{2}{9}\theta^2,\\
X_2&=0.
\end{align*}

Notice that the first equation is just a constraint on the evolution of $\theta$ (compare to the evolution equation \eqref{evo1} with vanishing shear):

\begin{eqnarray*}
\frac{2}{3}\dot{\theta}\neq\mathcal{A}\phi.
\end{eqnarray*}
Indeed it is clear that the spacetime cannot accelerate, and the above condition is simply the statement that the expansion is non-constant. In the irrotational and twisting LRS case, this is the only applicable CMC condition.

\item \(\hat{\chi}=\dot{\chi}=0\):

\ \\
In this case, we have that

\begin{eqnarray}\label{hiya1}
\left(\alpha_1^2+\alpha_2^2\right)X_2+\alpha_1\alpha_2\left(X_3+X_1\right)=0,
\end{eqnarray}
and hence, the CMC condition reduces to the pair of constraints on \(\mathcal{T}\):

\begin{align*}
X_3+X_1&=0,\\
X_2&=0.
\end{align*}
Since for $\xi\neq0$ implies $\alpha_2=0$ \cite{as5}, here  $\xi=0$, i.e. only LRS II type spacetimes allows for this CMC criterion (we henceforth consider those cases with vanishing cosmological constant). The first constraint above is, explicitly,

\begin{eqnarray}\label{expl1}
-\left(\rho+p+\Pi\right)=0.
\end{eqnarray}

We note a quite important statement here:\\

\textit{For any non-vacuum, accelerating, and dynamical GCS LRS spacetime, the mean curvature $\chi$ can be constant along neither $u^a$ nor $n^a$.}
\ \\

The above statement follows from the fact that if $\mathcal{A}\neq0$ (in which case the spacetime is necessarily of the LRS II type) and $\chi$ is constant along either one of $u^a$ nor $n^a$, it must be constant along the other by the CMC conditions. Because this requires \eqref{expl1} to hold on each leaf along $x^a$, there we must have $\dot{\rho}=0$. Since the spacetime is conformally static, $\mathcal{L}_x\rho=0$ implies $\hat{\rho}=0$ since $\alpha_2\neq0$ (for otherwise $\mathcal{A}=0$).

This is quite important as it implies that for for any GCS LRS spacetime which is not of the perfect fluid matter type, the timelike gradient CKV must have a nonvanishing component along $n^a$ as is seen from \eqref{cme3}. This then implies that all these spacetimes will fall into the next characterization. As will be seen in Section \ref{sec5}, this imposes quite stringent restrictions on horizon character in such spacetimes.

\item \(\hat{\chi}\neq0; \dot{\chi}\neq0\):

\ \\
For this case, we have that

\begin{eqnarray}\label{hiya2}
\left(\alpha_1^2-\alpha_2^2\right)X_2+\alpha_1\alpha_2\left(X_3-X_1\right)=0,
\end{eqnarray}
and the CMC condition reduces to the pair of constraints on \(\mathcal{T}\):

\begin{align*}
X_3-X_1&=0,\\
X_2&=0.
\end{align*}

Just as with the previous case, here, $\xi=0$. The first constraint above is then, explicitly,

\begin{eqnarray}\label{expl2}
-\frac{1}{3}R-2\mathcal{E}=0.
\end{eqnarray}

Now, it can be checked in a straightforward manner that the additional constraint \eqref{cme3}, i.e. the inequality specifying the timelike criterion, for this case, reduces to the statement (we have used the fact that \(X_3-X_1=0\) on \(\mathcal{T}\))

\begin{eqnarray}\label{expl3}
\varphi^2\left(\alpha_1^2-\alpha_2^2\right)\left[\varphi^2-\left(\alpha_1^2-\alpha_2^2\right)X_1\right]>0.
\end{eqnarray}
Since the difference \(\alpha_1^2-\alpha_2^2\) is strictly positive, it follows that the above inequality is equivalent to the requirement that

\begin{eqnarray}\label{expl4}
\varphi^2-\left(\alpha_1^2-\alpha_2^2\right)X_1>0.
\end{eqnarray}

Indeed, the condition \(X_1<0\) is sufficient for the inequality \eqref{expl4} to hold. In fact, this condition is quite reasonable on most physical grounds, and here is why: Suppose our interest is to work with the case of a vanishing cosmological constant \(\Lambda\). It is clear that there is a delicate interplay between the Weyl contribution and the scalar curvature of the ambient spacetime, on \(\mathcal{T}\). We know that if the scalar curvature \(R\) of the ambient spacetime is positive on \(\mathcal{T}\), it remains so along the conformal trajectory. If the electric Weyl scalar \(\mathcal{E}<0\), and one then insists that the pressure is non-negative, \(p\geq0\), it is easily checked that this will ensure that the SEC holds, a physically reasonable condition. It will then follow that \(X_1<0\).

Additionally, the following important comments are in order:

 \begin{itemize}

 \item The spacetime is scalar-flat, \(R=0\), if and only if it is conformally flat to conformal observers. This consequently implies that the spacetime must be flat as all thermodynamic variables would necessarily vanish.

 \item Conversely, a conformally flat LRS spacetime admits this criterion if and only if the conformal observers do not experience a non-zero scalar curvature.

 \end{itemize}

\end{enumerate}

We point out that in all of the cases discussed above, \(X_2\) is required to vanish. This is the criterion that the hypersurface neither gives off nor absorbs radiation. That is to say that the presence of heat flux is an obstruction to the existence of a timelike gradient CKV in the spacetimes of interest here. Stated more concretely, \textit{a conformally stationary LRS spacetime whose conformal observers experience a nonvanishing heat flux cannot be GCS.}

Now we are in the position to return to the result of Proposition \ref{theo1}. Consider a solution with a vanishing $Q$, and let us suppose that \eqref{imp1} is true. Furthermore, let us assume that $\alpha_i\neq0,\dot{\chi}\neq0,$ and $\hat{\chi}\neq0$. Then, from \eqref{cdi5} and \eqref{cme3}, it follows that one has the following ratio of the derivatives of the mean curvature function:

\begin{eqnarray}
\frac{\hat{\chi}}{\dot{\chi}}=-\frac{g_1}{g_2+\Pi}.\label{if1}
\end{eqnarray}

Suppose we have a perfect fluid matter type. If $p=0$, we have $g_1=g_2$. In this case $\alpha_1=\alpha_2$, i.e. $x^a$ is null, contradicting that $x^a$ is timelike. One immediately has that

\begin{proposition}
An inhomogeneous pressureless LRS perfect fluid is not gradient conformally stationary.
\end{proposition}

An immediate consequence of the above proposition is that the Lemaitre-Tolman-Bondi solution is not gradient conformally stationary.

Let us stay with the perfect fluid case, but suppose $p\neq0$. Applying the timelike criterion to \eqref{if1}, we find that if both the WEC and SEC are imposed, necessarily, the pressure is strictly negative, $p<0$: The timelike criterion imposes

\begin{eqnarray}
g_1^2-g_2^2=2p\left(\rho+p\right)<0.\label{if2}
\end{eqnarray}
It would then follow from the SEC that $\rho<0$. Hence, we state that

\begin{proposition}
An inhomogeneous LRS perfect fluid which obeys the weak and strong energy conditions is not gradient conformally stationary.
\end{proposition}

\begin{remark}[Remark 1.]
The above propositions appear to further strengthen the statement of Theorem \ref{theo1}. Explicitly stated, the Propositions above imply that \textit{the only GCS LRS perfect fluid obeying standard energy conditions is the Robertson-Walker type solution.}
\end{remark}


\section{Implications for the presence of black holes and character of horizons}\label{sec5}


In this section we discuss some implications of the CMC condition for the existence/presence of black holes and the character of their associated horizons. Specifically, for a particular foliation associated to one of the three CMC conditions, it is possible, in principle, to 1.) determine the evolution of a given marginally outer trapped surface and 2.) detect the presence of a black hole. The interested reader is referred to \cite{hamid1,as1,as2,as3} for more discussion on the character of horizons in LRS solutions.

As was earlier mentioned, the scalars $X_j$ are related to the evolution of null expansion scalars of 2-surfaces, which will be clarified shortly. We know that the evolution of the null expansion scalars of surfaces are used in the characterization of horizons of black holes. Therefore, we expect that the ratio of the linear combinations of the scalars $X_j$ which characterizes the CMC condition could be used to determine the causal character of horizons in spacetimes, and consequently may be used to determine whether or not a given horizon bounds a black hole. To help with easier following of the discussions and results of this section, we give a brief but decent introduction to the objects of interest here.

A compact 2-surface $\mathcal{S}$ in a spacetime admits two null directions, identifying the directions of outgoing and ingoing null rays. We denote these directions respectively by \(k^a\) and \(l^a\), and they are normalized as $k_al^a=-1$. For our purposes here, with specialization to LRS spacetimes, we consider surfaces of constant $t$ and $r$, i.e. those with induced metric $N_{ab}$ (we shall assume, henceforth, spherical symmetry). Then, the respective expansions (divergences) of $k^a$ and $l^a$ are

\begin{eqnarray*}
\theta_k=N^{ab}\nabla_ak_b;\qquad \theta_l=N^{ab}\nabla_al_b.
\end{eqnarray*}

A \textit{marginally trapped surface} (MTS) is that on which, at all points, \(\theta_k=0\) and \(\theta_l<0\). A 3-dimensional hypersurface foliated by MTS is called a \textit{marginally trapped tube} (MTT). (We follow some of the standard references on the subject \cite{sh1,ash1,ash2,ib1,ib2,ib3}.) In the case of LRS spacetimes, fixing the gauge $k^a=u^a+n^a$ and $l^a=(1/2)(u^a-n^a)$, the expansions have the respective covariant forms \cite{hamid1,as1}

\begin{eqnarray}
\theta_k=\frac{2}{3}\theta-\sigma+\phi;\qquad \theta_l=\frac{1}{2}\left(\frac{2}{3}\theta-\sigma-\phi\right).
\end{eqnarray}

On an MTT $\bar H$, one can always find a (constant) function $c$ such that the vector field \cite{sh1,ib2,ib3}

\begin{eqnarray}
y^a=k^a-cl^a,
\end{eqnarray}
is tangent to $\bar H$. (Variation of $\mathcal{S}$ along $y^a$ induces the foliation.) Since $y_ay^a=2c$, the causal character of $\bar H$ at a point is determined by the sign of $c$: spacelike for $c>0$, timelike for $c<0$, and null for $c=0$. To then determine an explicit expression for the function $c$, one simply notes that the expansion $\theta_k$ is Lie dragged along $y^a$, i.e. $\mathcal{L}_y\theta_k=0$, to find

\begin{eqnarray}
c=\frac{\mathcal{L}_k\theta_k}{\mathcal{L}_l\theta_k}\;.
\end{eqnarray}

It will be understood that $c$ is constant at all points of $\bar H$. In this case, $c$ characterizes the evolution of a given MTS. That is, an MTS will evolve into a spacelike, timelike, or null MTT if \(c\) is positive, negative, or zero, respectively.

If the NEC holds, $\mathcal{L}_k\theta_k\leq0$. It therefore follows that characterization of an MTS is captured in the sign of $\mathcal{L}_l\theta_k$, and the MTTs have specialized names: If the NEC strictly holds, the MTT is a dynamical horizon (DH) for $\mathcal{L}_l\theta_k<0$ and a timelike membrane (TLM) for $\mathcal{L}_l\theta_k>0$. The MTT is an isolated horizon (IH) (or simply null horizon if there is no ambiguity) if and only if $\mathcal{L}_k\theta_k=0$.

A \textit{marginally outer trapped surface} (MOTS) generalizes an MTS, where there is no sign restriction on the ingoing expansion $\theta_l$). Henceforth, we simply use `MOTS'.

Now, all of the above discussions with respect to the function $c$ are related to a notion of the \textit{stability} of MOTS which was introduced by Andersson \textit{et al.} in \cite{and1}. This notion is captured via an eigenvalue problem for a second order elliptic operator, obtained through the variation of the expansion $\theta_k$ along the field $n^a$ scaled by some positive function. Stability is then determined by the sign and `\textit{realness}' of the eigenvalue, with strict positivity implying strict stability. Roughly put, the work of \cite{and1} demonstrates that a strictly stable MOTS in an initial data set, under the assumption of the DEC, will evolve into a DH containing trapped surfaces just to the `inside'. Containing trapped surfaces just to the inside of a horizon is therefore akin to the condition that $\mathcal{L}_l\theta_k<0$, and hence, under the NEC assumption, only DH and IH will bound black holes (and obviously, on the DH for an evolving black hole).

It was observed in \cite{as4} that while the positivity of $c$, for dynamical horizons, is necessary for the stability of a MOTS (at least in the LRS II case), it is not sufficient. It is further required that the bound

\begin{eqnarray}\label{bound1}
0<c<1,
\end{eqnarray}
is obeyed, which imposes that the expansion $\theta_k$ must decrease along both the preferred temporal and spatial directions for a evolving dynamical black hole, i.e. $\dot{\theta}_k,\hat{\theta}_k<0$. This in fact provides an easy way to rule out stable MOTS:\\
\ \\
\noindent\textit{A MOTS with outgoing null expansion that is non-decreasing along at least one of the canonical directions $u^a$ or $n^a$ will not evolve into a DH containing a black hole.}\\
\ \\
So, for instance, for a given LRS II spacetime $M$, on the 2-sphere, one can immediately compute and check the sign of the temporal and spatial derivatives of $\theta_k$ to establish non-existence of spherical DH in $M$.

Now, since we hope to ensure our considerations capture truly dynamical black hole LRS spacetimes (by this we mean those cases where the horizons are non-minimal in the sense that, the null expansions do not both vanish identically), only the LRS II type solutions are admissible as was shown in \cite{as3}. We will therefore restrict to the LRS II class of spacetimes. For this case, explicitly, the variations of the outgoing null expansion along the two null directions are found as

\begin{align}
\mathcal{L}_k\theta_k&=-\left(\rho+p+\Pi\right)+2Q,\label{vex1}\\
\mathcal{L}_l\theta_k&=\frac{1}{3}\left(\rho-3p\right)+2\mathcal{E}.\label{vex2}
\end{align}

We have seen in the previous section that the conformal observers do not detect a non-vanishing $Q$, as this is an obstruction to the CMC condition. If one starts with a spacetime with a vanishing $Q$, one can analyze the horizon behavior without concerns for the discrepancy in flux observation by the conformal observers and those with four-velocity $u^a$. (It will be seen shortly that the CMC condition has a more direct consequence for MOTS evolution. This, in fact, can already be seen from the equations \eqref{vex1} and \eqref{vex2}.)  In the case of a vanishing $Q$ we will insist that the anisotropy scalar $\Pi$ is nonzero, since otherwise this would imply that the spacetime must be the Robertson-Walker solution, which we do not analyze here. (The reason will be adequately clarified at the end of the forthcoming subsection.)

On the other hand, for $Q\neq0$, the relationship between the CMC condition and MOTS dynamic depends crucially on whether the black hole radiates or absorbs radiation. As such, we will split our discussion on the implications of the CMC condition for presence and evolution of MOTS into two subcases: $Q=0$ and $Q\neq0$.

\subsection{Vanishing $Q$}

Let us begin with the case of vanishing $Q$. In the absence of the heat flux $Q$, it is clear that

\begin{align}
\mathcal{L}_k\theta_k&=X_3+X_1,\label{vex3}\\
\mathcal{L}_l\theta_k&=X_1-X_3,\label{vex4}
\end{align}
providing a direct relationship between the CMC condition and MOTS evolution. The bound on the constant function $c$ then imposes that $X_1,X_3<0$ on a DH. This provides an alternative non-existence check, just as with the case of the derivatives of the null expansion $\theta_k$. 

In fact, by expressing the variations of the null expansion $\theta_k$ in terms of the scalars $X_j$, one can establish the following bound on the Weyl scalar $\mathcal{E}$:

\begin{eqnarray}
\mathcal{E}<p.\label{vex5}
\end{eqnarray}
This follows from the fact that

\begin{eqnarray*}
X_1<0\quad\mbox{implies}\quad \frac{1}{2}\Pi>-\mathcal{E}-\frac{2}{3}\rho,
\end{eqnarray*}
so that

\begin{eqnarray*}
X_3<0\quad\mbox{implies}\quad \eqref{vex5}.
\end{eqnarray*}

For example, it is clear that a pressureless solution will necessarily have that $\mathcal{E}<0$, in the absence of the cosmological constant.

We could discuss several consequences for the spacetime variables, with regards to the form of the variations \eqref{vex3} and \eqref{vex4}. However, this is not the essence here. Our interest is in what is encoded in the scalars $X_j$ about the (non-)existence of black holes and their horizon properties. The role of the right hand side of \eqref{vex3} and \eqref{vex4}, in the CMC problem for LRS spacetimes, allows us to make several statements about existence of black holes and the causal character of horizons, which we write down as a collection of propositions. The NEC is assumed and the foliation induced by the CKV will be denoted by $\mathcal{W}_x$ as before.

\begin{proposition}\label{propos1}
Suppose an LRS II spacetime $M$ admits a timelike gradient (proper) CKV \(x^a\) and a leaf $\mathcal{T}$ of $\mathcal{W}_x$ contains a marginally outer trapped 2-sphere \(\mathcal{S}\). If the mean curvature of $\mathcal{T}$ is constant along $u^a$ and \(X_1<X_3\) on \(\mathcal{S}\), then the spacetime contains a black hole with a null boundary (IH). Consequently, the spacetime is singular.
\end{proposition}

To see the above, note that the condition $X_1<X_3$ ensures that $\mathcal{L}_l\theta_k<0$ (which then implies that the third CMC criterion is ruled out). The condition $\dot{\chi}=0$ also  then rules out the first criterion, leaving the second CMC criterion, so that the $S$ evolves into a null horizon with trapped surfaces just to the inside.

Alternatively, simply considering a shearing and/or accelerating spacetime rules out the first criterion.

We next state the following result:

\begin{proposition}\label{propos2}
Suppose an LRS II spacetime $M$ admits a timelike gradient (proper) CKV \(x^a\), and a leaf $\mathcal{T}$ of $\mathcal{W}_x$ has mean curvature $\chi$ that is nonconstant along $n^a$. Then, a horizon $\bar H$ in $M$ which intersects $\mathcal{T}$ at a marginally outer trapped 2-sphere \(\mathcal{S}\) cannot enclose a black hole.
\end{proposition}

Firstly, we note that by imposing that $\chi$ is nonconstant along the $n^a$ direction says that only the third CMC criterion is possible. And since the hypersurfaces $\mathcal{T}$ and $\bar H$ intersect at the surface $\mathcal{S}$, $\mathcal{L}_l\theta_k$ must vanish along $\bar H$.

Following from the above propositions, we also have as a corollary

\begin{corollary}\label{cor1}
Let an LRS II spacetime $M$ admit a gradient (proper) CKV \(x^a\). Then, a dynamical horizon in $M$ cannot intersect a leaf of $\mathcal{W}_x$.
\end{corollary}

The following interpretation of the result of Corollary \ref{cor1} is understood: \textit{The region in which a timelike gradient CKV is defined can admit neither an isolated nor a dynamical horizon.}

Corollary \ref{cor1} says that, for a black hole enclosed by a dynamical horizon, an observer with conformal motion does not `observe' the black hole changing in area. Such a black hole would be in equilibrium as per the observers with conformal motion, even though the $u^a$ - observers do experience the evolution of the black hole. The corollary is seen to be somewhat related to Propositions 6.2 and 6.3 of the reference \cite{ash2}  except that that was the Killing case (see also the similar conclusions drawn in \cite{mm1}).

Also, it is now evident that for a GCS LRS spacetime, the behavior of the mean curvature $\chi$ of the spacelike hypersurfaces along the spatial vector $n^a$ plays a crucial role in determining the existence of a black hole enclosing horizon. In particular, whether $\chi$ is constant along $n^a$ provides some immediate information on the presence of trapped surfaces and the evolution of MOTS. We briefly mention aspects of this case-wise:

\begin{enumerate}

\item $\hat{\chi}=0$: In the case of a non-vanishing acceleration (we consider this non-vanishing $\mathcal{A}$ scenario to discard the Robertson-Walker case), if a leaf of $\mathcal{W}_x$ contains a MOTS, the MOTS will evolve to a null horizon. Trapped surfaces will be enclosed provided that the scalar and conformal (electric) Weyl curvatures satisfy $R<-6\mathcal{E}$. On the other hand, there will be no enclosure of trapped surfaces if the reverse inequality holds.

\item $\hat{\chi}\neq0$: In this case, if a horizon intersects a leaf of $\mathcal{W}_x$, the horizon does not enclose trapped surfaces. This is due to the fact that the leaves of $\mathcal{W}_x$ cannot contain a stable MOTS. Otherwise, if there is no intersection with $\mathcal{W}_x$, independent of the development of the associated black hole, the black hole will be seen as being in equilibrum by conformal observers.

\end{enumerate}

It is quite important to emphasize that as we set $Q=0$ from the onset, $\Pi\neq0$ (as was earlier mentioned) due to Theorem \ref{theo1} (as well as Remark 1.), since $\Pi=0$, a perfect fluid, would imply the spacetime is Robertson-Walker type. Whether or not there is a black hole depends on whether the model is an expanding or collapsing one (the latter which allows for the presence of trapped surfaces), a dichotomy we will not pursue further here. In other words, the scalar $\Pi$ in the GCS case, plays a role in sourcing the dynamics that lead to the formation of trapped surfaces (and consequently the bounding horizon enclosing the trapped sufaces).

\subsection{Nonvanishing $Q$}

Let us briefly comment on the case of a nonvanishing heat flux. When discussing the previous results of this section, we had simply considered the case with a vanishing heat flux to entertain the possibility of the intersection of a black hole horizon with a leaf of $\mathcal{W}_x$. However, $Q$ will generally not be zero. It is then clear that the character of the relationship between $X_j$ and the evolution of the null expansions becomes a bit more intricate. More specifically, there is now a difference in the numerator of $c$ with the explicit form

\begin{eqnarray}
\mathcal{L}_k\theta_k=X_3+X_1+2Q.\label{nvq1}
\end{eqnarray}

For simplicity, we will consider the case where $X_3+X_1=0$ vanishes on the horizon, so that we have a horizon condition compatible with observers with a (gradient) conformal flow. It is quite clear that if $X_3+X_1=0$, then the NEC demands that $Q\leq0$, with the equality case already covered in the previous subsection. Thus, the case of a non-vanishing $Q$ requires, necessarily, that $Q<0$. That is, such a horizon will experience an outward flux. Therefore, unless there is flow of some other matter across the horizon from the `outside', if we assume that the horizon area growth is dictated by $Q$, the horizon area must be decreasing and hence a timelike membrane.

Of course, depending on the dynamics of the interacting fields of the spacetime to the exterior of the black hole, different horizon characters are possible. In any case, for the horizons considered here, foliated by constant $t-r$ MOTS, they will not intersect $\mathcal{W}_x$, and conformal observers will again only see a black hole in equilibrum since there is a non-zero flux. While there is no intersection, we simplify our discussion by allowing for $X_1+X_3$ to vanish on the horizon, as this is also perceived in the case of the second CMC criterion. In this way, the difference experienced between the conformal observers and the $t$-observers are characterized via the flux $Q$. So, for example, for the timelike scenario described above, the conformal observers would not observe the collapse of such a black hole.


\section{Discussion}\label{sec6}


The existence of gradient conformal Killing vectors plays a special role in spacetime decomposition. In the case that the CKV is timelike, the spacetime splits along the CKV. These spacetimes are the so-called gradient conformally stationary spacetimes. (In the case of LRS spacetimes the associated vorticity of the CKV vanish identically and there the `stationary' can be replaced with `static'. (See comments on this in \cite{as5}).) Their existence in fluid spacetimes has been considered in \cite{nd1}, and more recently \cite{as5} where the class of LRS spacetimes was considered, with the latter reference constructing a wave-like PDE whose solutions are in bijection to gradient CKVs. The approach in \cite{as5} was a direct analysis of the conformal Killing equations, a purely kinematic consideration,written down in covariant form and imposing the vanishing of the conformal bivector. In this work, which in some aspects supplements the results of those carried out in the references \cite{nd1} and \cite{as5}, we have approached this problem more geometrically. We obtained a relationship between components of the Ricci tensor along the CKV, and the gradient of its divergence, which introduces the matter variables in a pair of first order covariant equations for the potential function of the CKV. These equations were analyzed for the existence of a gradient CKV, with the integrability condition being provided. Several interesting observations were made. It was also established that for a perfect fluid LRS spacetime with positive energy density and a nonvanishing pressure, the Robertson-Walker type solution is the unique solution admitting a timelike gradient CKV, thereby providing the unique GCS solution for perfect fluids. In the case of a nonvanishing heat flux $Q$, the equation to be analyzed for the existence of a timelike gradient CKV was obtained.

If a spacetime admits a timelike gradient CKV, the CKV provides a choice of privileged observers, which introduces a foliation of the spacetime by spacelike hypersurfaces in the timelike region. The trace of the extrinsic curvature (or second fundamental form) of these spacelike hypersurfaces, referred to as the mean curvature, is constant with respect to the induced covariant derivative on the hypersurface. Therefore, establishing the constant mean curvature (CMC) conditions provides necessary conditions for the existence of timelike gradient CKVs. We have established the CMC conditions for LRS spacetimes. These are given by three distinct conditions, with a GCS spacetime obeying exactly one of these conditions. The three conditions are all characterized by combinations of a set of three covariant scalars, notated $X_j$ for $j\in\lbrace{1,2,3\rbrace}$, which are comprised of the spacetime curvature variables. In each case, we discussed some of the implications for the spacetime variables. For example, the energy conditions on the spacelike hypersurfaces play a crucial role in determining whether the CMC conditions hold. In some cases, a very delicate interplay between the scalar and electric Weyl curvatures will characterize which one of the CMC conditions hold. The three cases share a common requirement that there is no flux across the hypersurfaces, i.e. on the hypersurfaces $Q=0$. This immediately presents an obstruction result: a conformal Killing observer that observes a non-zero heat flux is not gradient.

Additionally, the CMC condition allows us to strengthen the uniqueness result for perfect fluids, previously obtained, by relaxing the nonvanishing condition on the pressure. This essentially rules out, for example, the Lemaitre-Tolman-Bondi type solution as GCS, or more generally, any inhomogeneous perfect fluid that obeys the the standard energy conditions.

Finally, we discussed some of the implications of the CMC conditions for the existence/presence of black holes and the causal character of associated horizons. It turns out that the particular combinations of the scalars $X_j$, whose vanishing characterize the CMC conditions, have a very close relationship with the evolution of the expansion of outgoing null geodesics from constant $t-r$ surfaces, along outgoing and ingoing null directions (in the absence of heat flux $Q=0$ this relationship is in fact coincident). In this case, depending on the relationship between the $X_j$, we obtained results which describe the presence of trapped surfaces and MOTS in the spacetime region where the CKV is defined, and which bear implications for whether or not a MOTS contained in a leaf of the foliation $\mathcal{W}_x$ can be stable. The various results accumulate in the statement that \textit{if an LRS II spacetime admits a timelike gradient CKV $x^a$, a dynamical horizon in the spacetime cannot intersect a leaf of the induced foliation $\mathcal{W}_x$}. In other words, the timelike region of the CKV can admit no dynamical horizon.

In the case that $Q\neq0$, clearly a horizon will not intersect $\mathcal{W}_x$ as these observers do not experience a nonzero $Q$. We can however assess how the conformal observers interpret the evolution of a black hole with a vanishing condition of the horizon coinciding with one of the CMC conditions. In this case, the null energy condition imposes that there is outgoing flux across the horizon. If we assume that the flux $Q$ is the sole determinant in the area growth of the horizon, then this would imply that the horizon area is decreasing, which in turn implies that the horizon is timelike. However, to the conformal observers, the area of the black hole is non-changing since they do not observe a nonvanishing $Q$. In other words, a collapse of such a black hole will not be observed by these observers. This is a frame problem similar to (but not nearly the same) the switching between Schwarzschild isotropic and Eddington-Finkelstein coordinates. Another parallel is seen in the case of conformal Killing horizons \cite{dy1}, where a dynamical black hole is enclosed by a null hypersurface, as seen by the conformal observers, even though the usual fluid observers with four-velocity will see a non-static black hole (see the references \cite{ab1,tarafdar2023slowly,as6} for discussions of this problem in the case of the Vaidya metric and its charged counterpart). Of course, there are subtleties to be taken care of when evoking such comparisons, but we will not go into detail here.

To conclude, a potential future research direction, which does not directly extend the current work but is nonetheless inextricably linked, is to consider the existence of trapped surfaces and MOTS in the presence of general spacetime symmetries. In particular, without restricting proper conformal symmetries to the gradient case, under which conditions are trapped surfaces and MOTS allowed in a spacetime, or at least in the region where such fields are defined? Also, to what extent do some results for the true Killing case, for example in \cite{mm1,ash2}, extend to the proper conformal Killing case? We will also aim to examine the relationship of such considerations to the stability of MOTS. It would be nice to approach this quite generally rather than restricting to the LRS class of spacetimes. For example, the 1+1+2 decomposition seems quite natural to study this problem since 2-surfaces (possibly distorted) with induced metric $N_{ab}$ have the canonical normal defined by the splitting. All these we will seek to address in a subsequent paper.

\section*{Data availability statement}

No new data were created or analyzed in this study.

\section*{Acknowledgement}

AS research is supported by the Basic Science Research Program through the National Research Foundation of Korea (NRF) funded by the Ministry of Education (grant numbers) (NRF-2022R1I1A1A01053784) and (NRF-2021R1A2C1005748). PKSD acknowledges support through the First Rand Bank, South Africa. AS would also like to thank the IBS Center for Theoretical Physics of the Universe, Daejeon, South Korea, for its hospitality, where a part of this work was carried out.

\section*{Appendix: Covariant LRS field equations}

We write down the covariant field equations in terms of the 1+1+2 variables  (see \cite{cc1,cc2} for further details).

\begin{itemize}

\item \textit{Evolution}

\begin{subequations}
\begin{align}
\frac{2}{3}\dot{\theta}-\dot{\sigma}&=\mathcal{A}\phi-\frac{1}{2}\left(\frac{2}{3}\theta - \sigma\right)^2+2\Omega^2-\frac{1}{3}\left(\rho+3p-2\Lambda\right)+\mathcal{E}-\frac{1}{2}\Pi,\label{evo1}\\
\dot{\phi}&=\left(\frac{2}{3}\theta-\sigma\right)\left(\mathcal{A}-\frac{1}{2}\phi\right)+2\xi\Omega+Q,\label{evo100}\\
\dot{\Omega}&=\mathcal{A}\xi-\left(\frac{2}{3}\theta-\sigma\right)\Omega,\label{aevo1}\\
\dot{\xi}&=-\frac{1}{2}\left(\frac{2}{3}\theta-\sigma\right)\xi+\frac{1}{2}\mathcal{H}+\left(\mathcal{A}-\frac{1}{2}\phi\right)\Omega,\label{aevo2}\\
\dot{\mathcal{E}}-\frac{1}{3}\dot{\rho}+\frac{1}{2}\dot{\Pi}&=-\left(\frac{2}{3}\theta-\sigma\right)\left(\frac{3}{2}\mathcal{E}+\frac{1}{4}\Pi\right)+\frac{1}{2}\phi Q+\frac{1}{2}\left(\rho+p\right)\left(\frac{2}{3}\theta-\sigma\right)+3\xi\mathcal{H},\label{evo101}\\
\dot{\mathcal{H}}&=-3\xi\mathcal{E}-\frac{3}{2}\left(\frac{2}{3}\theta-\sigma\right)\mathcal{H}+\Omega Q+\frac{3}{2}\xi\Pi.\label{aevo3}
\end{align}
\end{subequations}

\item \textit{Propagation}

\begin{subequations}
\begin{align}
\frac{2}{3}\hat{\theta}-\hat{\sigma}&=\frac{3}{2}\phi \sigma+2\xi\Omega +Q,\label{evo2}\\
\hat{\phi}&=\left(\frac{1}{3}\theta+\sigma\right) \left(\frac{2}{3}\theta-\Sigma\right)-\frac{1}{2}\phi^2+2\xi^2-\frac{2}{3}\left(\rho+\Lambda\right)-\mathcal{E}-\frac{1}{2}\Pi,\label{evo200}\\
\hat{\Omega}&=\left(\mathcal{A}-\phi\right)\Omega,\label{aevo4}\\
\hat{\xi}&=-\phi\xi-\left(\frac{1}{3}\theta+\sigma\right)\Omega,\label{aevo5}\\
\hat{\mathcal{E}}-\frac{1}{3}\hat{\rho}+\frac{1}{2}\hat{\Pi}&=-\frac{3}{2}\phi\left(\mathcal{E}+\frac{1}{2}\Pi\right)+3\Omega\mathcal{H}-\frac{1}{2}\left(\frac{2}{3}\theta-\sigma\right)Q,\label{evo201}\\
\hat{\mathcal{H}}&=-\left(3\mathcal{E}+\rho+p-\frac{1}{2}\Pi\right)\Omega-\frac{3}{2}\phi\mathcal{H}-Q\xi.\label{aevo6}
\end{align}
\end{subequations}

\item \textit{Propagation/Evolution}

\begin{subequations}
\begin{align}
\hat{\mathcal{A}}-\dot{\theta}&=-\left(\mathcal{A}+\phi\right)\mathcal{A}+\frac{1}{3}\theta^2+\frac{3}{2}\sigma^2+\frac{1}{2}\left(\rho+3p-2\Lambda\right),\label{evo3}\\
\hat{Q}+\dot{\rho}&=-\theta\left(\rho+p\right)-\left(\phi+2\mathcal{A}\right)Q-\frac{3}{2}\sigma\Pi,\label{evo300}\\
\hat{p}+\hat{\Pi}+\dot{Q}&=-\left(\frac{3}{2}\phi+\mathcal{A}\right)\Pi-\left(\frac{4}{3}\theta+\sigma\right)Q-\left(\rho+p\right)\mathcal{A},\label{evo301}
\end{align}
\end{subequations}
as well as the constraint

\begin{eqnarray}\label{eevvoo}
\mathcal{H}=3\sigma\xi-\left(2\mathcal{A}-\phi\right)\Omega.
\end{eqnarray}
with \(\Lambda\) being the cosmological constant. The scalar \(\mathcal{H}=H_{ab}e^ae^b\) is the magnetic part of the Weyl tensor, which vanishes for the cases that are of interest in the present work.
\end{itemize}

\end{document}